\documentclass{article}

\usepackage{amstext,amssymb,amsmath,amsthm}
\usepackage{url,color}
\usepackage{hyperref}

\theoremstyle{definition}
\newtheorem{theorem}{Theorem}[section]

\newtheorem{corollary}[theorem]{Corollary}
\newtheorem{proposition}[theorem]{Proposition}
 
\newtheorem{definition}[theorem]{Definition}

\newcommand{\lang}{{\mathcal{L}}}
\newcommand{\FinDFA}{{\textsf{FinDFA}}}
\newcommand{\CofinDFA}{{\textsf{CofinDFA}}}
\newcommand{\DFA}{{\textsf{DFA}}}
\newcommand{\NFA}{{\textsf{NFA}}}
\newcommand{\CFG}{{\textsf{CFG}}}
\newcommand{\NP}{{\textsf{NP}}}

\title{Quick Brown Fox in Formal Languages}
\author{Kazuhiro Inaba (\texttt{kiki@kmonos.net})}

\begin{document}

\maketitle
\begin{abstract}
Given a finite alphabet $\Sigma$ and a deterministic finite automaton on $\Sigma$,
the problem of determining whether the language recognized by the automaton contains
any {\em pangram} is \NP-complete.
Various other language classes and problems around pangrams are analyzed.
\end{abstract}


\section{Introduction}

\begin{definition}
A string $w \in \Sigma^*$ is a {\em pangram} ({\em perfect pangram}, respectively)
on finite alphabet $\Sigma$ if all $\sigma \in \Sigma$ appear at least once (exactly once) in $w$.
We denote the set of all pangrams on $\Sigma$ by $P_\Sigma$ and the set of
all perfect pangrams by $E_\Sigma$.
\end{definition}

The famous English pangram ``The quick brown fox jumps over the lazy dog''
is in $P_\Sigma$ for $\Sigma = \{\mathtt{a}, \ldots, \mathtt{z}\}$,
while it is not in $E_\Sigma$ because some of the letters appear more than once.
In Japanese, there is an old poem called Iroha-Uta\footnote{\url{https://en.wikipedia.org/wiki/Iroha}}
being a perfect pangram over 47 letters of Japanese alphabet.
In the programming language Ruby, a perfect pangram on ASCII printable characters
can be a valid program that outputs each ASCII printable character exactly
once.\footnote{\url{https://github.com/tric/trick2013/tree/master/kinaba}}

No one will doubt that coming up with the artistic sentence of the ``quick brown fox''
 or Iroha-Uta is a tough brain teaser.
The purpose of this paper is to confirm the natural intuition;
creating a pangram is indeed hard. We discuss on the computational
complexity of the problem deciding whether there is a (perfect) pangram in a given language.


\section{Problem Definitions}
We assume the reader to be familiar with basic concepts of formal languages
and computational complexity, e.g., by~\cite{HandbookFL,NPCBook}.

\begin{definition}
A {\em language acceptor class} $X$ consists of a sets $X_\Sigma$ of {\em language acceptors}
for each finite alphabet $\Sigma$ and a membership relation
$\mathit{mem}_X \in \bigcup_\Sigma X_\Sigma \times \Sigma^*$.
\end{definition}

Examples of language acceptor class are \DFA{} (deterministic finite automata),
\NFA{} (deterministic finite automata), \CFG{} (context free grammars), etc.
Each member of the set $\mathsf{DFA}_{\{0,1\}}$ is a deterministic finite automaton
over the binary alphabet.
We assume that each member $x \in X_\Sigma$ has a naturally associated size $|x|$,
such as the number of transitions or grammar rules.
For $x \in X_{\Sigma}$, we write $\lang(x)$ for $\{w \in \Sigma^* \mid \mathit{mem}_X(x, w)\}$.

\begin{definition}[Pangram Problem]
{\em Pangram problem} (respectively, {\em perfect pangram problem}) for a language acceptor class $X$ is,
given a finite alphabet $\Sigma$ and a language acceptor $x \in X_{\Sigma}$ as an input,
to ask whether $\lang(x) \cap P_\Sigma$ (resp. $\lang(x) \cap E_\Sigma$) contains some member or not.
\end{definition}

Informally speaking, the pangram problem gives us a grammar of a language and asks if
there can be a pangram in the language.
In the rest of the paper, we measure the computationally complexity of (perfect) pangram
problems in terms of $|x| + |\Sigma|$.

Here are some easy facts on pangrams, from the state complexity perspective:
\begin{proposition}
$P_\Sigma$ is a regular language.
The minimum deterministic finite automata recognizing $P_\Sigma$ has $2^{|\Sigma|}$ states.
\end{proposition}

\begin{proposition}
$E_\Sigma$ is a regular language (actually a finite language.)
The minimum deterministic finite automata recognizing $E_\Sigma$ has $2^{|\Sigma|}+1$ states.
\end{proposition}

\begin{proof}
Count the number of Myhill-Nerode equivalence classes.
\end{proof}

These exponential natures block us from giving efficient algorithms
by a naive construction for many decision problems on pangrams,
as we see in the following sections.


\section{Hardness of Pangram Problems}

Firstly, the problem is \NP-hard already for \DFA.

\begin{theorem} \label{thm:dfa-e-nph}
Perfect pangram problem for \DFA{} is \NP-hard.
\end{theorem}
\begin{proof}
The proof is by reduction from Hamiltonian Path Problem, which asks,
given a directed graph $(V, E)$, if there is a simple path visiting
all the nodes exactly once.

The instance of Hamiltonian Path Problem can be converted to a perfect
pangram problem for \DFA{} as follows.
We let the alphabet $\Sigma = V$ and
the set of states of the automaton to be $Q = V \cup \{q_\text{src}, q_\text{fail}\}$.
The initial and the final states of the automaton is
$q_\text{src}$ and $Q \setminus \{q_\text{fail}\}$,
respectively. We let the transition function 
\begin{align*}
  \delta(q_\text{src}, u) &= u \text{ for all } u \in V \\
  \delta(v, u) &= u \text{ for all } v,u \in V \text{ where } (v, u) \in E\\
  \delta(v, u) &= q_\text{fail} \text{ for all } v,u \in V \text{ where } (v, u) \notin E\\
  \delta(q_\text{fail}, u) &=  q_\text{fail} \text{ for all } u \in V
\end{align*}
Then, if there is a Hamiltonian path $v_1 \rightarrow \cdots \rightarrow v_n$,
there exists a pangram $v_1 \cdots v_n$ in the language accepted by the automaton,
and vice versa.
\end{proof}

\begin{theorem} \label{thm:dfa-p-nph}
Pangram problem for \DFA{} is \NP-hard.
\end{theorem}
\begin{proof}
Reduction from the perfect pangram problem for \DFA.
Note that for each $\Sigma$, the set $S_\Sigma = \{w \mid |w|=|\Sigma|\}$
is a regular language and there exists a deterministic finite automaton $s_\Sigma$
for the language, with size polynomial in $\Sigma$.

Then, given a problem instance $(\Sigma, x)$ of the perfect pangram problem,
$\lang(x)$ contains a perfect pangram
if and only if $\lang(x) \cap \lang(s_\Sigma)$ contains a pangram,
since $\lang(s_\Sigma)$ fixes the string length and makes all pangrams also perfect pangrams.
Computing the deterministic finite automaton representing the
the intersection is well-known to be done in polynomial time by the product construction.
\end{proof}

The automaton constructed in the proof of Theorem~\ref{thm:dfa-p-nph} represents
a {\em finite} language. In other words, the \NP-hardness arises even from such a
restricted class of languages.

\begin{proposition}
Perfect pangram problem and pangram problem are \NP-hard for \FinDFA,
where \FinDFA{} is the subclass of \DFA{} representing only finite languages.
\end{proposition}

Interestingly, a language whose complement
is finite (called {\em cofinite} languages) also exhibits the same hardness on
the perfect pangram problem.
Note that the normal pangram problem for \CofinDFA{} is trivial; a cofinite language
always contains a pangram because there are infinitely many pangrams.

\begin{proposition}
Perfect pangram problem is \NP-hard for \CofinDFA,
where \CofinDFA{} is the subclass of \DFA{} representing only cofinite languages.
\end{proposition}
\begin{proof}
Reduction from the perfect pangram problem of general \DFA{}.
Assume a DFA $g$ is given, and
let $e_\Sigma$ be a DFA representing the set $\{w \mid |w|\neq|\Sigma|\}$,
which can be constructed in $|\Sigma|+1$ states.
Then, $\lang(g)$ contains a perfect pangram if and only if $\lang(g) \cup \lang(e_\Sigma)$ does.
The latter is a cofinite language.
\end{proof}

Now, let us turn our eyes to the containment in \NP{}, to complete the 
proof of the \NP-completeness.

\begin{theorem} \label{thm:cfg-e-np}
Perfect pangram problem is in \NP{} for \CFG.
\end{theorem}
\begin{proof}
The witness, i.e. a perfect pangram string, has length $\Sigma$ and can be guessed by
a nondeterministic Turing machine in linear time. Then, since the combined complexity
of checking the membership of the string against the acceptor is in polynomial time,
the whole step of checking perfect pangram problem is in \NP.
\end{proof}

The result extends to classes with higher expressiveness as long as they
permit \NP-time membership judgment, but we need some care.
The membership check needs to be in \NP{} in combined complexity,
meaning that both the string and the language acceptor are counted as inputs.

For instance, although the membership problem for the higher order
version of \CFG{} known as IO/OI-hierarchy~\cite{Damm82} is in \NP~\cite{IM08},
the complexity is only with respect to the string size and a grammar is considered to be fixed.
If we take the grammar as a part of input, the membership problem is
\textsf{EXP}-complete~\cite{TK86} already for indexed languages~\cite{Aho68,Fischer68a},
which is only one level above \CFG.
It remains \textsf{PSPACE}-complete even after limiting to the
subclass of grammars without $\epsilon$-rules~\cite{Osterholzer15}.

Compared to perfect pangrams, a pangram may have arbitrary length.
It may not be easy to find a short witness directly. For instance,
there exists a context-free grammar whose shortest member string is exponentially
larger than the grammar. Nevertheless, the problem is proved to stay
in \NP{} by applying a simple grammar transformation.

For preparation,
let's call a string $u$ is a {\em subsequence} of a string $w$ and write $u \sqsubseteq w$
if $u$ is obtained by
deleting several letters from $w$ without changing the order of the remaining letters.
The {\em downward closure} of a language $L$ is the set of strings
  $L_\downarrow = \{u \mid u \sqsubseteq w \text{ for some } w \in L\}$.
The important fact is that $L$ contains a pangram if and only if
$L_\downarrow$ contains a perfect pangram.

\begin{theorem} \label{thm:cfg-p-np}
Pangram problem is in \NP{} for \CFG.
\end{theorem}
\begin{proof}
Pangram problem reduces to perfect pangram problem as long as the downward
closure for \CFG{} can be constructed in polynomial time--and indeed it is.
Downward closure of \CFG{} can be computed in polynomial time,
for instance by constructing Chomsky normal form and then adding a rule $N \rightarrow \epsilon$
for all nonterminals $N$.
\end{proof}

Putting altogether and considering that conversion from \DFA{} to \CFG{} is trivially
done in polynomial time, we get the following summary of this section.

\begin{corollary}
Perfect pangram problem is
 \NP-complete for \FinDFA, \CofinDFA, \DFA, \NFA, and \CFG.
Pangram problem is
 \NP-complete for \FinDFA, \DFA, \NFA, and \CFG.
\end{corollary}


\section{Pangram-Cover Problems}

Viewing from the other side, asking if a language contains some pangram is
equivalent to ask if the complement of the language contains all pangrams,
and to negate the answer. Thus we obtain the following yet another hardness statement.

\begin{corollary}~\label{cor:e-conp}
Given a language acceptor $g$, the problem asking whether
$\lang(g) \supseteq E_\Sigma$ is \textsf{coNP}-complete for
\FinDFA, \CofinDFA, \DFA, \NFA, and \CFG.
\end{corollary}
\begin{proof}
It is in \textsf{coNP} because the negative witness can be given by a non-member
perfect pangram, which is in linear size.
It is \textsf{coNP}-hard because the negation of this problem is equivalent to the
perfect pangram problem for the complement of $g$. Since the perfect pangram problem
is \NP-hard for \FinDFA, \CofinDFA, and \DFA{}, for the language acceptors
that can subsume the polynomial-time constructible complement of either of the three classes,
the result follows.
\end{proof}

The pangram counterpart of this ``contains all'' question becomes truly harder
for \NFA{} and \CFG{}, which are not polynomial-time closed under complement.

\begin{theorem}
Given a language acceptor $g$, the problem asking whether
$\lang(g) \supseteq P_\Sigma$ is \textsf{coNP}-complete for \CofinDFA{} and \DFA,
{\em undecidable} for \CFG, and \textsf{PSPACE}-complete for \NFA.
\end{theorem}
\begin{proof}
The \textsf{coNP}-hardness follows by the same argument as Corollary~\ref{cor:e-conp}.
The \textsf{coNP} containment for the \CofinDFA{} and the \DFA{} cases holds because
for DFA we can still always find a short negative witness by a slight modification
of the well-known {\em pumping lemma}.

The undecidability for \CFG{} is proved by the reduction from the universality problem,
which asks whether $\lang(g) = \Sigma^*$ or not. This problem is known to be undecidable~\cite{GR63}.
Now, assume a CFG $g$ over alphabet $\Sigma = \{\sigma_1, \ldots, \sigma_n\}$ is given.
Let $w = \sigma_1\cdots\sigma_n$, $w\lang(g)$ the language obtained by
concatenating $w$ in front of each string in $\lang(g)$, and $\mathit{NS}_w$ the set of strings
that do {\em not} start with $w$. Then it is easy to verify:
\begin{align*}
       & \lang(g) = \Sigma^*
  \\
  \iff & w\lang(g) \cup \mathit{NS}_w = \Sigma^*
  \\
  \iff & w\lang(g) \cup \mathit{NS}_w \supseteq P_\Sigma.
\end{align*}
Since, $\mathit{NS}_w$ is a context free language and \CFG{} is closed under
concatenation and union, the decision procedure of the last inequality
will decide the universality problem as well, which is impossible.

The \textsf{PSPACE}-hardness for \NFA{} is similarly derived from the universality
problem (known to be \textsf{PSPACE}-complete for \NFA~\cite{MS72,Hunt73})
by using the fact that $\mathit{NS}_w$
is representable in polynomial-size \NFA{} and so for the union.

The containment in \textsf{PSPACE} can be proved by the backward reduction.
Notice that the set $\overline{P_\Sigma}$ of non-pangrams
can be represented by a polynomial-size NFA by seeing the language as
 $\bigcup_{a \in \Sigma} (\Sigma \setminus \{a\})*$. Thus, by testing the
universality of $\lang{g} \cup \overline{P_\Sigma}$ we can decide whether
$\lang{g} \supseteq P_\Sigma$ or not in \textsf{PSPACE}.
\end{proof}


\section{On Strictly Locally Testable Languages}

It should be natural to ask if we can find a weaker class of language acceptors
that admits tractable (perfect) pangram problems. The problems were hard even for
finite languages, but they are not the only sub-regular class.

Our first choice of the trial is {\em local languages}~\cite{Medvedev64},
or equivalently {\em 2-strictly locally testable languages}~\cite{MP71,BS73}.
For an integer $k \geq 2$, $k$-strictly locally testable languages ($k$-slt for short) are
represented by three sets $S \subseteq \Sigma^{k-1}$,
$I \subseteq \Sigma^k$, and $E \subseteq \Sigma^{k-1}$.
A string $w$ belongs to the $k$-slt language denoted by $(S,I,E)$ if
and only if $w$ begins with a prefix in $S$, ends with a suffix in $E$,
and all $k$ consecutive substrings in $w$ are in $I$.
We call $k$\textsf{-SLT} the class of $k$-slt language acceptors whose size is measured by
the sum of representing set sizes $|S|+|I|+|E|$.\footnote{Since we can convert this representation
to a minimal DFA in polynomial time, the \NP-hardness results in this section
still hold even if we assume the $k$\textsf{-SLT} is given by a DFA.}

Unfortunately, the perfecet pangram problem still remains to be \NP-hard.

\begin{proposition}
Perfect pangram problem is \NP-hard for \textsf{2-SLT}.
\end{proposition}
\begin{proof}
The same proof as Theorem~\ref{thm:dfa-e-nph} applies, because the represented
language is actually 2-slt, with $S=E=\Sigma$ and edges denote the permitted neighbors $I$.
\end{proof}

Contrary to the perfect pangram problem, the normal pangram problem becomes
tractable when it is limited to \textsf{2-SLT}.
The product construction proof of Theorem~\ref{thm:dfa-p-nph}
does not apply to \textsf{2-SLT} because the resulting new language
is not in 2-slt anymore.

\begin{theorem}
Pangram problem for \textsf{2-SLT} can be solved in linear time.
\end{theorem}
\begin{proof}
Let $(S,I,E)$ be a tuple representing a local language,
and consider a graph with $\Sigma$ is the set of nodes and
$I$ the set of edges. Whether or not a pangram is contained in the
language is equivalent to ask whether there is a (not necessarily simple)
path beginning from $S$ and ending in $E$ that visits all nodes $\Sigma$.

This graph-theoretic problem is solvable through the decomposition to strongly-connected components.
A path visiting all nodes in the original graph exists if and only
if the acyclic graph obtained by contracting strongly-connected components
has a Hamiltonian path from a component containing a $S$ node to a component with $E$.
For an acyclic graph, the Hamiltonian path (if any) is uniquely obtained by
topologically ordering the nodes. All these condition and checkable in linear time.
\end{proof}

If we go one step higher to the class of \textsf{3-SLT},
the pangram problem becomes hard again.

\begin{theorem}
Pangram problem is \NP-hard for \textsf{3-SLT}.
\end{theorem}
\begin{proof}
Again the proof is by reduction from Hamiltonian Path Problem.
Let us assume the given instance of Hamiltonian Path Problem is $(V, E)$.
We let the alphabet $\Sigma = V \cup \{1, \ldots, |V|\}$,
and the corresponding instance $(S_l,I_l,E_l)$ of \text{3-SLT} is
constructed as follows:
\begin{align*}
  S_l &= \{(v,1) \mid v \in V\} \\
  I_l &= \{(v, k, u) \mid (v,u) \in E, k \in \{1, \ldots, |V|\}\} \\
      &\phantom{=} \cup \{(k, v, k+1) \mid v \in V, k \in \{1, \ldots, |V|-1\}\} \\
  E_l &= \{(v, |V|) \mid v \in V\}
\end{align*}
Then, if there is a Hamiltonian path $v_1 \rightarrow \cdots \rightarrow v_n$
there exists a pangram $v_1 1 v_2 2 \cdots (n-1) v_n n$ in the language accepted by the automaton,
and vice versa.
\end{proof}


\section{On Strictly Piecewise Testable Languages}

Beside the hierarchy of locally testable languages, yet another well-studied
classes of sub-regular languages is those of {\em piecewise testable languages}~\cite{Simon75}.
Recall that by $u \sqsubseteq w$ we mean $u$ to be a (not necessary contiguous)
subsequence of $w$. A language $L$ is $k${\em-strictly piecewise}~\cite{Heinz07,RHBEVWW10}
if $L$ can be defined in terms of forbidden subsequences, i.e., if $L$ can be written as
$L = \{w \mid \forall u \in F. u \not\sqsubseteq w\}$ for some $F \subseteq \Sigma^{\leq k}$.
Let $k$\textsf{-SPT} be the class of $k$-strictly piecewise language acceptors
whose size is measured by the size $|F|$ of the forbidden subsequence set.

First of all, since all strictly piecewise languages are downward closed,
it contains a pangram if and only if it contains a perfect pangram.
Hence the two problems become the same.
For a lower class of this hierarchy, the problems become tractable.

\begin{proposition}
Perfect pangram problem and pangram problem can be solved in linear time
for \textsf{2-SPT}.
\end{proposition}
\begin{proof}

If $F$ contains an empty string or a string of length-1, the language
trivially excludes pangrams. Hence we assume $F \subseteq \Sigma^2$ below.

Note that each member $(x,y)$ in $F$ represents the constraint
that the letter $x$ must come {\em after} $y$ in a pangram (otherwise $x$
comes before $y$ in the pangram, which contradicts $F$.)

Now, consider a graph with nodes $\Sigma$ and edges $F$. If this graph
has a cycle, no pangram is contained in the language because all pangram
must violate the ``must come after'' constraint for some edge in the cycle.
If this graph has no cycle, then the reverse topological sort will give
a pangram, in which all $x$ comes after $y$ for all $(x,y) \in F$ by the
definition of topological sorting. Cycle detection of a directed graph
can be done in linear time.
\end{proof}

Not surprisingly, once we change the parameter from 2 to 3,
the \NP-hardness creeps in.

\begin{theorem}
Perfect pangram problem and pangram problem are \NP-complete for \textsf{3-SPT}.
\end{theorem}
\begin{proof}
It is in \NP{} because the witness of perfect pangram is in linear size.

The reduction for \NP-hardness is from the Betweenness Problem~\cite{Opatrny79},
that asks there is a total ordering of elements of a finite set $\Sigma$, under the set of
given constraints $(a,b,c) \in \Sigma^3$ that imposes either $a < b < c$ or $c < b < a$
(i.e., $b$ must be in between the other two.)

We can construct the forbidden set of \textsf{3-SPT} over the same $\Sigma$.
For each betweenness constraint $(a,b,c)$, the forbidden subsequences ``$acb$'', ``$cab$'',
``$bac$'', and ``$bca$'' are added. Then, if a total ordering exists, the element of
$\Sigma$ listed along the order gives a perfect pangram over $\Sigma$ avoiding all forbidden
subsequences, and vice versa.
\end{proof}

Interesting open problem is the complexity when the \textsf{3-SPT} language is specified
in the form of a DFA. The minimum DFA representation can be exponentially larger than
the forbidden set representation, hence our hardness result does not directly apply.


\section{Tractability of Always-Pangram Problems}

Another interesting question on a language acceptor is to decide
whether all the strings it accepts are (perfect) pangrams.
Despite the \NP-hardness of showing existence of one pangram,
checking if {\em all} of them are pangram is efficiently decidable.

\begin{theorem}
For a \CFG{} $g$ over alphabet $\Sigma$, whether or not $\lang(g) \subseteq P_\Sigma$
is decidable in $O(|\Sigma|\cdot|g|)$ time.
\end{theorem}
\begin{proof}
For each $\sigma \in \Sigma$, compute the intersection of $g$ and
$(\Sigma \setminus \{\sigma\})^*$.
This is obtained by just dropping all the rules containing $\sigma$ in right-hand side.
If the obtained grammar recognizes a non-empty language (which is linear time decidable),
there is a non-pangram member in $\lang(g)$.
If all of the intersections are empty, $\lang(g) \subseteq P_\Sigma$ holds.
\end{proof}

\begin{corollary}
For a \CFG{} $g$ over alphabet $\Sigma$, whether or not $\lang(g) \subseteq E_\Sigma$
is decidable in $O(|\Sigma|\cdot|g|)$ time.
\end{corollary}
\begin{proof}
$\lang(g) \subseteq E_\Sigma$ is equivalent to saying that 
$\lang(g) \subseteq P_\Sigma$ and all members of $\lang(g)$ is of length $|\Sigma|$.
The latter condition is easily checkable in $O(|g|)$ time,
by assigning the length for each nonterminal
in a bottom-up manner (from the ones with right-hand side consisting of terminals) and
checking no contradiction.
\end{proof}

The problems analyzed in this section have natural interpretation when the language is considered
to be describing some sequence of events possibly happening.
The condition $\lang(g) \subseteq P_\Sigma$
corresponds to say that all events will eventually happen for all possible
event sequences.

\section{Related Work}
The All Colors Shortest Path problem (ACSP)~\cite{BCGSAE15}
asks the shortest path in an undirected graph with nodes colored,
under the constraint that the paths must visit all colors.
The pangram problem for \NFA{} can be regarded as an edge-colored and directed
version of All Colors Path problem. Our result shows that the problem is \NP-complete
even if the shortestness condition is dropped and just the existence of such paths is asked.

\section{Conclusion}

\begin{table}[t]
\small
\begin{tabular}{r|c|c|c|c|c|c}
      & $\lang \cap P_\Sigma \neq \emptyset$
      & $\lang \cap E_\Sigma \neq \emptyset$
      & $\lang \supseteq P_\Sigma$
      & $\lang \supseteq E_\Sigma$
      & $\lang \subseteq P_\Sigma$
      & $\lang \subseteq E_\Sigma$
\\
\hline
  \textsf{2-SLT}    & \textsf{P}   & \textsf{NPc} & - & - & \textsf{P} & \textsf{P}
\\
  \textsf{3-SLT}    & \textsf{NPc} & \textsf{NPc} & - & - & \textsf{P} & \textsf{P}
\\
  \textsf{2-SPT}    & \textsf{P}   & \textsf{P}   & - & - & - & -
\\
  \textsf{3-SPT}    & \textsf{NPc} & \textsf{NPc} & - & - & - & -
\\
  \textsf{CofinDFA} & -            & \textsf{NPc} & \textsf{coNPc}   & \textsf{coNPc} & - & -
\\
  \textsf{FinDFA}   & \textsf{NPc} & \textsf{NPc} & -                & \textsf{coNPc} & \textsf{P} & \textsf{P}
\\
  \DFA              & \textsf{NPc} & \textsf{NPc} & \textsf{coNPc}   & \textsf{coNPc} & \textsf{P} & \textsf{P}
\\
  \NFA              & \textsf{NPc} & \textsf{NPc} & \textsf{PSPACEc} & \textsf{coNPc} & \textsf{P} & \textsf{P}
\\
  \CFG              & \textsf{NPc} & \textsf{NPc} & undecidable      & \textsf{coNPc} & \textsf{P} & \textsf{P}
\end{tabular}
\caption{Summary of the results (- means the problem is almost trivial.)}
\label{tbl:summary}
\end{table}

We have defined the notion of pangram and perfect pangram in terms of formal language theory.
Computational complexity of several problems around pangrams are investigated.
The result presented in this paper is summarized in Table~\ref{tbl:summary}.


\bibliographystyle{alpha}
\bibliography{../kinaba,sub}

\newcommand{\etalchar}[1]{$^{#1}$}
\begin{thebibliography}{B{\c{C}}G{\etalchar{+}}15}

\bibitem[Aho68]{Aho68}
Alfred~V. Aho.
\newblock Indexed grammars--an extension of context-free grammars.
\newblock {\em Journal of the ACM}, 15:647--671, 1968.

\bibitem[B{\c{C}}G{\etalchar{+}}15]{BCGSAE15}
Yunus~Can Bilge, Dogukan {\c{C}}agatay, Beg{\"{u}}m Gen{\c{c}}, Mecit Sari,
  H{\"{u}}seyin Akcan, and Cem Evrendilek.
\newblock All colors shortest path problem.
\newblock \url{http://arxiv.org/abs/1507.06865}, 2015.

\bibitem[BS73]{BS73}
J.A. Brzozowski and Imre Simon.
\newblock Characterizations of locally testable events.
\newblock {\em Discrete Mathematics}, 4:243--271, 1973.

\bibitem[Dam82]{Damm82}
Werner Damm.
\newblock The {IO}- and {OI}-hierarchies.
\newblock {\em Theoretical Computer Science}, 20:95--207, 1982.

\bibitem[Fis68]{Fischer68a}
Michael~J. Fischer.
\newblock Grammars with macro-like productions.
\newblock In {\em Switching and Automata Theory}, pages 131--142, 1968.

\bibitem[GJ79]{NPCBook}
Michael~R. Garey and David~S. Johnson.
\newblock {\em Computers and Intractability; A Guide to the Theory of
  NP-Completeness}.
\newblock W. H. Freeman \& Co., New York, NY, USA, 1979.

\bibitem[GR63]{GR63}
Seymour Ginsburg and Gene~F. Rose.
\newblock Some recursively unsolvable problems in {ALGOL}-like languages.
\newblock {\em Journal of the ACM}, 10:29--47, 1963.

\bibitem[Hei07]{Heinz07}
Jeffrey Heinz.
\newblock {\em The Inductive Learning of Phonotactic Patterns}.
\newblock PhD thesis, University of California, Los Angeles, 2007.

\bibitem[{III}73]{Hunt73}
H.~B.~Hunt {III}.
\newblock On the time and tape complexity of languages {I}.
\newblock In {\em ACM Symposium on Theory of Computing (STOC)}, 1973.

\bibitem[IM08]{IM08}
Kazuhiro Inaba and Sebastian Maneth.
\newblock The complexity of tree transducer output languages.
\newblock In {\em Foundations of Software Technology and Theoretical Computer
  Science (FSTTCS)}, pages 244--255, 2008.

\bibitem[Med64]{Medvedev64}
Yu~T. Medvedev.
\newblock On the class of events representable in a finite automaton ({English}
  translation).
\newblock In {\em Sequential Machines: Selected Papers}. Addison-Wesley, 1964.

\bibitem[MP71]{MP71}
Robert McNaughton and Seymour~A. Papert.
\newblock {\em Counter-Free Automata (M.I.T. Research Monograph No. 65)}.
\newblock The MIT Press, 1971.

\bibitem[MS72]{MS72}
A.~R. Meyer and L.~J. Stockmeyer.
\newblock The equivalence problem for regular expressions with squaring
  requires exponential space.
\newblock In {\em Switching and Automata Theory}, pages 125--129, 1972.

\bibitem[Opa79]{Opatrny79}
J.~Opatrny.
\newblock Total ordering problem.
\newblock {\em SIAM Journal on Computing}, 8:111--114, 1979.

\bibitem[Ost15]{Osterholzer15}
Johannes Osterholzer.
\newblock Complexity of uniform membership of context-free tree grammars.
\newblock In {\em International Conference on Algebraic Informatics (CAI)},
  pages 176--188, 2015.

\bibitem[RHB{\etalchar{+}}10]{RHBEVWW10}
James Rogers, Jeffrey Heinz, Gil Bailey, Matt Edlefsen, Molly Visscher, David
  Wellcome, and Sean Wibel.
\newblock On languages piecewise testable in the strict sense.
\newblock In {\em The Mathematics of Language (MOL)}, pages 255--265, 2010.

\bibitem[RS97]{HandbookFL}
Grzegorz Rozenberg and Arto Salomaa, editors.
\newblock {\em Handbook of Formal Languages, Vol. 1: Word, Language, Grammar}.
\newblock Springer-Verlag New York, Inc., New York, NY, USA, 1997.

\bibitem[Sim75]{Simon75}
Imre Simon.
\newblock Piecewise testable events.
\newblock In {\em Automata Theory and Formal Languages}, pages 214--222, 1975.

\bibitem[TK86]{TK86}
Shinichi Tanaka and Takumi Kasai.
\newblock The emptiness problem for indexed language is exponential-time
  complete.
\newblock {\em Systems and Computers in Japan}, 17:29--37, 1986.

\end{thebibliography}
\end{document}